\documentclass[lettersize,journal]{IEEEtran}
\usepackage{amsmath,amsfonts}
\usepackage{algorithmic}
\usepackage{algorithm}
\usepackage{graphicx}
\usepackage[caption=false, font=footnotesize]{subfig}
\usepackage{array}
\usepackage{textcomp}
\usepackage{stfloats}
\usepackage{url}
\usepackage{verbatim}
\usepackage{cite}
\usepackage{xcolor}
\usepackage{multirow}
\usepackage{lipsum}
\usepackage{siunitx}
\usepackage{balance}
\usepackage{amssymb}
\usepackage{amsthm}
\theoremstyle{plain}
\newtheorem{theorem}{Theorem}
\newtheorem{definition}{Definition}

\newtheorem{lemma}{Lemma}

\newtheorem{problem}{Problem}[section]

\newtheorem{remark}{Remark}

\newcommand{\cidx}{c} % Constraint index for CBFs
\begin{document}

% \title{Distributed Perception-Aware Safe Control for 3D leader--follower Formation via CBF}
\title{Distributed 3D Leader–Follower Formation Control with Field-of-View Safety via Control Barrier Functions}
\author{Immanuel R. Santjoko,$^{1\dagger}$
 Richie R. Suganda,$^{1\dagger}$
 Miao Pan,$^{1}$
 and Bin Hu$^{1,2}$%
 \thanks{$^\dagger$The first two authors contributed equally to this work.}%
 \thanks{$^{1}$I. R. Santjoko, R. R. Suganda, B. Hu, and M. Pan are with the Department of Electrical and Computer Engineering, University of Houston, Houston, TX 77004, USA. ({\tt\small isantjoko@uh.edu, rrsugand@uh.edu, bhu11@central.uh.edu})}%
 \thanks{$^{2}$B. Hu is with the Department of Engineering Technology, University of Houston, Houston, TX 77004, USA.}%
}

% \author{Anonymous Authors for Double-Blind Review}

\markboth{Journal of \LaTeX\ Class Files,~Vol.~14, No.~8, August~2021}%
{Shell \MakeLowercase{\textit{et al.}}: A Sample Article Using IEEEtran.cls for IEEE Journals}
\maketitle
\begin{abstract}
This letter proposes a distributed 3D leader–follower formation~(3D-LFF) control framework for multi-UAV systems that achieves formation tracking while enforcing perception safety constraints. Maintaining safe, vision-based 3D-LFF is challenging because onboard cameras impose strict Field-of-View (FOV) limitations, and demanding formation commands can drive the leader outside the follower’s camera frustum, resulting in loss of visibility. To address this issue, we develop a perception-aware safe control architecture that guarantees visibility by construction. First, we derive a relative kinematic model in a line-of-sight coordinate representation and design a distributed 3D-LFF tracking controller using only locally available relative states. Next, we embed the nominal formation controller within a Control Barrier Function–based Quadratic Program (CBF-QP) safety filter that minimally modifies the commanded velocities to maintain the leader inside the follower’s camera frustum while preserving formation tracking whenever feasible. Gazebo simulations and Crazyflie hardware experiments validate the proposed approach, demonstrating accurate formation tracking and effective FOV enforcement, including scenarios in which the nominal desired formation conflicts with visibility constraints.
\end{abstract}

% \begin{IEEEkeywords}
% Aerial systems: perception and autonomy, multi-agent systems, safety
% \end{IEEEkeywords}


\section{Introduction}\label{sec:intro}
\IEEEPARstart{U}{nmanned} aerial vehicles (UAVs) are increasingly deployed in inspection \cite{11170465} and exploration \cite{liu2022large} due to their agility and ability to access hazardous or hard-to-reach areas. 
To scale sensing coverage and mission efficiency, UAVs are often operated as multi-agent systems, where leader--follower formations provide a simple and scalable coordination primitive \cite{sun2016event, sun2017event}. 
In GPS-denied environments, however, formation performance depends critically on the follower's ability to estimate the leader's relative state using onboard sensing (e.g., monocular/stereo cameras, fiducial markers, or other relative sensors). 
This introduces an intrinsic \emph{perception constraint}: the leader must remain inside the follower's camera frustum to maintain continuous visibility and avoid degradation or loss of relative-state estimation.
As illustrated in Fig.~\ref{fig:intro}, this creates a direct conflict between formation objectives (tracking a desired 3D relative configuration) and sensing limitations (finite field of view (FOV) and depth range). 
Motivated by this challenge, we propose a distributed perception-aware control framework that relies only on \emph{local relative-state information} to achieve 3D leader--follower formation tracking while explicitly enforcing camera-frustum visibility constraints.

Vision-based leader--follower formation control broadly spans (i) \emph{position-/translation-based} formulations and (ii) \emph{pose-/bearing-based} formulations.
Early position-based approaches~\cite{cowan2003vision, panagou2014cooperative-2D} design feedback laws using relative translation (or reduced relative states), which can be effective for planar settings but typically do not model camera frustum constraints in a way that yields rigorous visibility guarantees.
More recent works incorporate richer relative information---e.g., bearing measurements, relative orientation cues, or fiducial-marker-based pose estimation---to improve robustness and enable more principled reasoning about sensing constraints. 
For example, marker-based relative localization has been used to support leader--follower flight in GPS-denied environments, demonstrating practical feasibility and sensitivity to delay and occlusions~\cite{walter2019uvdar-3Dboresight, abdelaal2021visualcollab}. 
Complementary approaches study bearing-only or bearing-dominant formation strategies and robustness to leader motion~\cite{ramirez2024vision}.
\begin{figure}[t]
 \centering
 \includegraphics[width=1\linewidth]{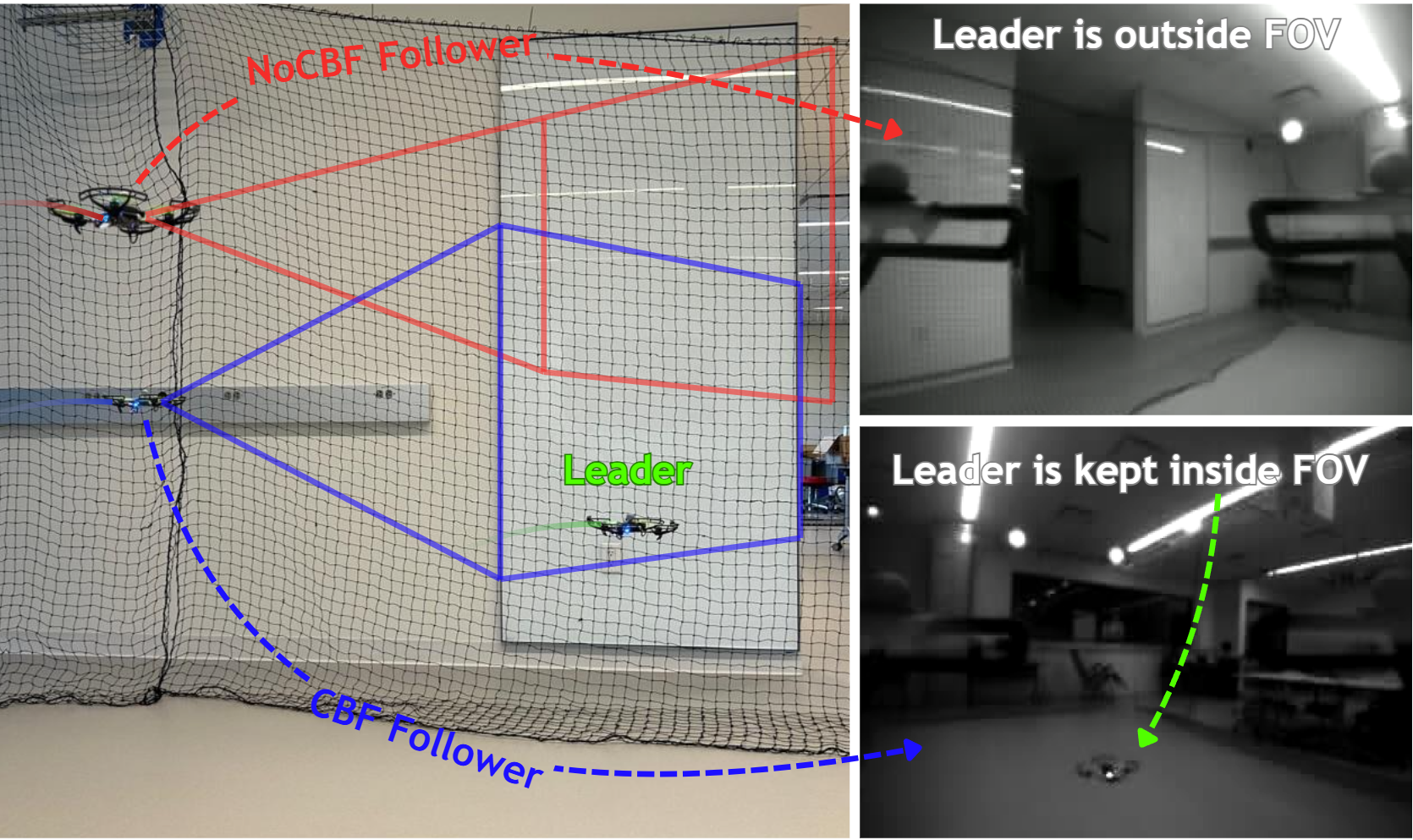}
 % \caption{Perception-aware leader–follower safe control. Leader–follower formations use UAVs with body-fixed cameras and limited FOV, requiring the leader to remain within the follower’s view for safe operation.}
 \caption{Perception-aware leader--follower safe control. Without perception constraints, the baseline follower (NoCBF, red) loses sight of the leader during maneuvers. In contrast, the proposed CBF-equipped follower (CBF, blue) actively modifies its commanded velocities to ensure the leader remains strictly within its camera's FOV.}
  \label{fig:intro}
\end{figure}
Perception constraints and visibility maintenance have also been studied in visual servoing and perception-aware control. 
A variety of methods enforce visibility/FOV constraints using constrained optimization, MPC, or barrier-like certificates; representative examples include visibility-constrained visual servoing and target tracking under camera limitations~\cite{salehi2021constrained, ning2024bearing, heshmati2024control}. 
Control Barrier Functions (CBFs) have emerged as a particularly attractive tool for safety-critical control due to their forward invariance guarantees and compatibility with real-time quadratic programs (CBF-QPs), with applications including cruise control \cite{hu2021trust}, obstacle avoidance \cite{zhang2025gcbf+}, and lane keeping \cite{bena2023safety}. 
Recent perception-aware CBF formulations address visibility constraints for tracking and collision avoidance~\cite{bena2023safety, kim2025visibility, trimarchi2025control}. 
Within multi-agent settings, distributed CBF-QP implementations and scalable safe control frameworks have been developed, including perception-aware leader--follower formulations for ground robots~\cite{suganda2025distributed, suganda2026formation} alongside other distributed methods~\cite{tan2021distributed, zhang2023neural}. Yet, directly extending these planar visibility methods to a \emph{3D camera frustum} in UAV formations remains nontrivial due to the coupled 3D geometry and relative heading effects.

This letter presents a distributed perception-aware framework for 3D leader--follower formation with explicit camera-frustum safety. The main contributions are three fold. Our first contribution is a relative 3D kinematic model in spherical coordinates that captures the camera-to-leader geometry (including relative heading) and yields a 3D leader--follower formation (3D-LFF) control law that relies on local relative measurements and distributed leader states. 
Second contribution is a real-time CBF-QP safety filter that enforces the follower's camera frustum---FOV and depth limits---by minimally modifying the nominal formation commands, thereby guaranteeing continuous visibility while preserving tracking whenever feasible. 
Third, the proposed method is validated in Gazebo simulations and Crazyflie experiments, demonstrating reliable formation tracking and strict FOV enforcement, including scenarios where the desired formation is incompatible with visibility constraints. 
% A video demonstrating these results is available at \url{https://youtu.be/cK3tFkGduiA}.}

The remainder of this paper is organized as follows. 
Section~II presents preliminaries. 
Section~III details the system model and problem formulation. 
Section~IV presents the proposed distributed controller and CBF-QP safety filter. 
Section~V reports Gazebo simulation and Crazyflie experimental results. 
Section~VI concludes the paper and discusses future directions.
\section{Preliminaries}\label{sec:preliminaries}

\textit{Notation}: Let $\mathbb{R}$, $\mathbb{R}_{\ge0}$, and $\mathbb{R}_{>0}$ denote the sets of real, non-negative, and positive real numbers, respectively. 
Let $SO(3)$ denote the special orthogonal group of $3 \times 3$ rotation matrices.
For $\mathbf{x}\in\mathbb{R}^n$, $\|\mathbf{x}\|$ denotes the Euclidean norm. 
The standard basis vectors of $\mathbb{R}^3$ are denoted by $\mathbf{b}_1,\mathbf{b}_2,\mathbf{b}_3$.
For a vector $\boldsymbol{\omega}\in\mathbb{R}^3$, $[\boldsymbol{\omega}]_{\times}$ denotes the associated skew-symmetric matrix. 
% For a scalar function $h(\mathbf{x})$, $\nabla_{\mathbf{x}} h$ denotes its gradient vector
A continuous function $\kappa:(-a,b)\rightarrow\mathbb{R}$ with $a,b\in\mathbb{R}_{>0}$ is an \emph{extended class-$\mathcal{K}$} function if it is strictly increasing and satisfies $\kappa(0)=0$.

Consider the control-affine system
\begin{equation}
\dot{\mathbf{x}} = f(\mathbf{x}) + g(\mathbf{x})\mathbf{u},
\label{eq:affine_sys}
\end{equation}
where $\mathbf{x}\in\mathcal{X}\subset\mathbb{R}^n$ is the state within the admissible state space $\mathcal{X}$, and $\mathbf{u}\in\mathcal{U}\subset\mathbb{R}^m$ is the control input drawn from the set of admissible controls $\mathcal{U}$. 
The functions $f:\mathcal{X}\to\mathbb{R}^n$ and $g:\mathcal{X}\to\mathbb{R}^{n\times m}$ are assumed locally Lipschitz on $\mathcal{X}$. Safety is encoded as forward invariance of a \emph{safe set} $\mathcal{S}\subset\mathcal{X}$ defined by a continuously differentiable function $h:\mathcal{X}\to\mathbb{R}$:
\begin{equation}
\mathcal{S} = \{\mathbf{x}\in\mathcal{X}\mid h(\mathbf{x})\ge 0\}.
\end{equation}
The system \eqref{eq:affine_sys} is safe if any trajectory starting in $\mathcal{S}$ remains in $\mathcal{S}$ for all $t\ge 0$.

\begin{definition}[Control Barrier Function {\cite{ames2016control}}]
A continuously differentiable function $h$ is a valid CBF for \eqref{eq:affine_sys} if there exists an extended class-$\mathcal{K}$ function $\kappa(\cdot)$ such that, for all $\mathbf{x}\in\mathcal{S}$,
\begin{equation}
\sup_{\mathbf{u}\in\mathcal{U}}
\left\{
\nabla_{\mathbf{x}}h(\mathbf{x})^\top\big(f(\mathbf{x})+g(\mathbf{x})\mathbf{u}\big) + \kappa\!\left(h(\mathbf{x})\right)
\right\}\ge 0.
\label{eq:cbf}
\end{equation}
Any feedback law $\mathbf{u}(\mathbf{x})\in\mathcal{U}$ satisfying \eqref{eq:cbf} renders $\mathcal{S}$ forward invariant for \eqref{eq:affine_sys}.
\end{definition}
\section{System Framework and Problem Formulation}\label{sec:system_framework}

This paper considers a multi-quadrotor UAV system where each vehicle is modeled as a rigid body, with $j$ and $i$ indexing the leader and follower, respectively. 
Let $\mathcal{F}_I$ denote an inertial frame, while $\mathcal{F}_{B_k}$ and $\mathcal{F}_{\psi_k}$ denote the body-fixed and yaw-aligned frames attached to UAV $k \in \{i, j\}$.
The dynamics of the $k$-th vehicle are given by
\begin{subequations}\label{eq:uav_dynamics}
\begin{align}
\dot{\mathbf{p}}_k &= \mathbf{v}_k, \label{eq:pos_dyn} \\
m_k \dot{\mathbf{v}}_k &= - m_k g \mathbf{b}_3 + T_k \mathbf{R}_k \mathbf{b}_3, \label{eq:vel_dyn} \\
\dot{\mathbf{R}}_k &= \mathbf{R}_k [\boldsymbol{\omega}_k]_\times, \label{eq:rot_dyn} \\
\mathbf{J}_k \dot{\boldsymbol{\omega}}_k &= \boldsymbol{\tau}_k - \boldsymbol{\omega}_k \times (\mathbf{J}_k \boldsymbol{\omega}_k), \label{eq:ang_dyn}
\end{align}
\end{subequations}
where $\mathbf{p}_k, \mathbf{v}_k \in \mathbb{R}^3$ denote the position and velocity expressed in $\mathcal{F}_{I}$, $\mathbf{R}_k \in SO(3)$ is the rotation matrix from $\mathcal{F}_{B_k}$ to $\mathcal{F}_{I}$, and $\boldsymbol{\omega}_k \in \mathbb{R}^3$ is the body angular velocity.
The constant $m_k \in \mathbb{R}_{>0}$ denotes the mass, $g \in \mathbb{R}$ the gravitational acceleration, and $\mathbf{J}_k \in \mathbb{R}^{3 \times 3}$ the inertia matrix. 
To drive this system, $T_k \in \mathbb{R}_{\geq 0}$ and $\boldsymbol{\tau}_k \in \mathbb{R}^3$ are defined as the actuation inputs, representing the total thrust and body torque, respectively.

% Let $\mathbf{p}_k\in\mathbb{R}^3$ and $\mathbf{R}_k\in SO(3)$ denote the position and attitude of UAV $k\in\{i,j\}$ in $\mathcal{F}_I$. 

For this leader--follower pair, let the subscript $(\cdot)_{ij}$ denote quantities of the leader $j$ relative to the follower $i$, as depicted in Fig.~\ref{fig:leader--follower-geometry}. 
% Let $\mathcal{F}_I$ be the inertial frame and $\mathcal{F}_{B_i}$ the body-fixed frame of follower $i$. 
The follower carries a forward-facing camera rigidly mounted along its body $x$-axis, with optical center offset $d_i\in\mathbb{R}_{>0}$ from the body origin. 
The relative position of the leader with respect to the follower camera, denoted by $\mathbf{q}_{ij} \triangleq [x_{ij},y_{ij},z_{ij}]^\top\in  \mathbb{R}^3$ and expressed in $\mathcal{F}_{B_i}$, is
\begin{equation}\label{eq:q}
\mathbf{q}_{ij} \triangleq 
\mathbf{R}_i^{\top}\left(\mathbf{p}_j-\mathbf{p}_i\right) - d_i\mathbf{b}_1.
\end{equation}
% where $\mathbf{b}_1$ is the first standard basis vector of $\mathbb{R}^3$.

To specify a 3D leader--follower formation using local measurements, we represent $\mathbf{q}_{ij}$ in spherical coordinates and augment it with a relative heading variable (Fig.~\ref{fig:leader--follower-geometry}). Define the relative state
\begin{equation}\label{eq:xij_def}
\mathbf{x}_{ij}\triangleq
\begin{bmatrix}
L_{ij} & \phi_{ij} & \xi_{ij} & \alpha_{ij}
\end{bmatrix}^{\top}\in  \mathbb{R}^4,
\end{equation}
where $L_{ij}\triangleq \|\mathbf{q}_{ij}\|$, $\phi_{ij}\in(-\pi,\pi]$ is the horizontal angle (azimuth), and $\xi_{ij}\in(-\pi,\pi]$ is the elevation angle, all expressed in $\mathcal{F}_{B_i}$. The relative heading is $\alpha_{ij}\triangleq \psi_j-\psi_i-\phi_{ij}$, where $\psi_i,\psi_j\in(-\pi,\pi]$ denote the follower and leader yaw angles defined relative to the inertial frame $\mathcal{F}_I$, respectively.

\begin{figure}[!t]
 \centering
 \includegraphics[width=1\linewidth]{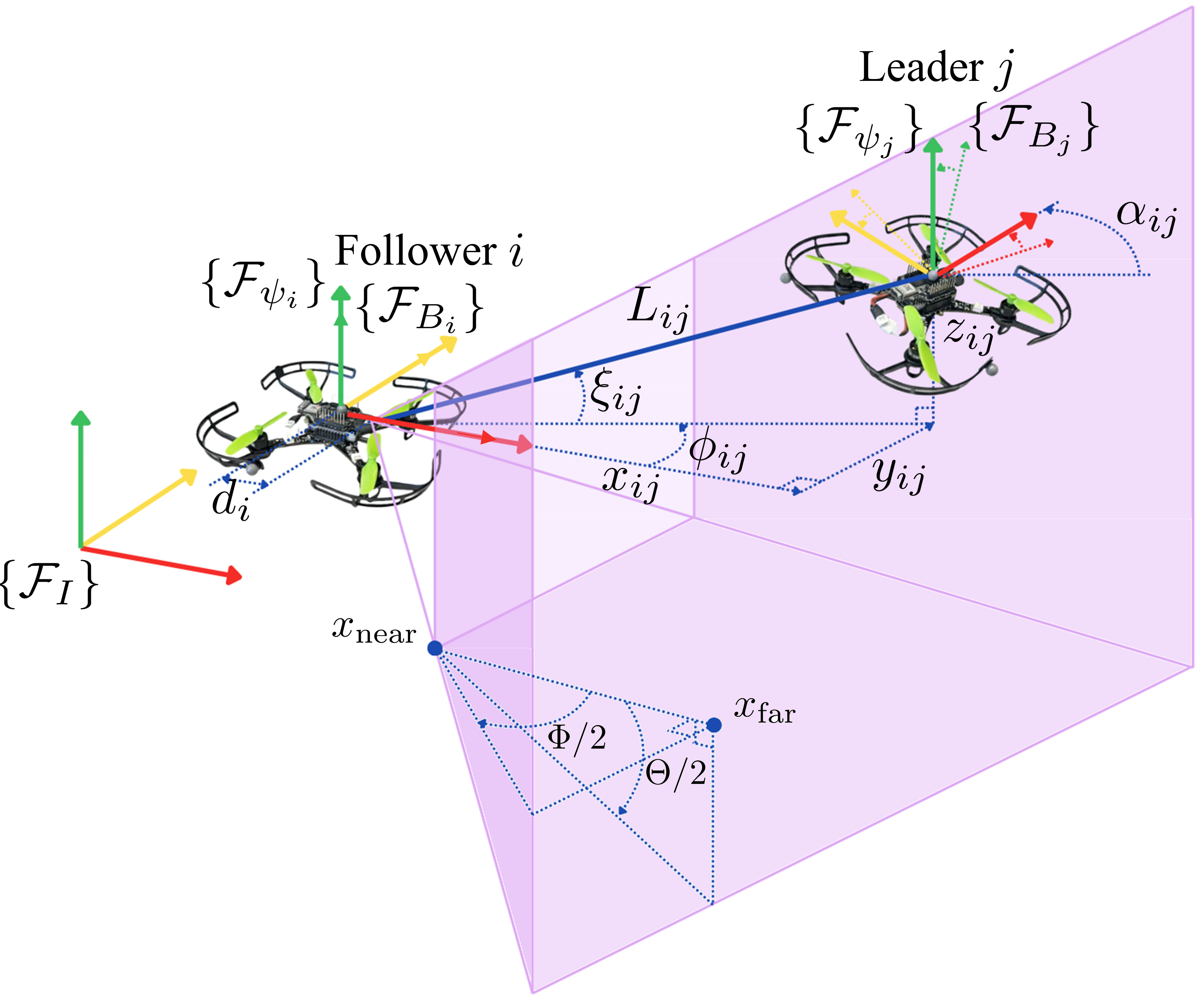}
 \caption{Leader--follower geometry and camera frustum. The leader position relative to follower $i$’s camera is expressed in $\mathcal{F}_{B_i}$ as $\mathbf{q}_{ij}=[x_{ij},y_{ij},z_{ij}]^\top$ and equivalently by $(L_{ij},\phi_{ij},\xi_{ij})$, augmented with the relative heading $\alpha_{ij}$. The magenta region depicts the camera frustum defined by depth bounds $[x_{\text{near}},x_{\text{far}}]$ and FOV angles $(\Phi,\Theta)$.
}
 \label{fig:leader--follower-geometry}
\end{figure}
Visibility is enforced by constraining the leader to lie inside the follower camera frustum. Let $\Phi, \Theta\in(0,\pi)$ denote the camera horizontal and vertical FOV angles, and let $x_\mathrm{near}, x_\mathrm{far}\in\mathbb{R}_{>0}$ denote the depth range limits along the camera optical axis, satisfying $x_\mathrm{near}<x_\mathrm{far}$.
We define the perception safe set $\mathcal{S}_{ij}^\mathbf{q}$ as the set of relative Cartesian positions satisfying  the frustum constraints (illustrated in Fig.~\ref{fig:leader--follower-geometry}):
% \begin{equation}\label{eq:safe_set_placeholder}
% \mathcal{S}_{ij}^\mathbf{q}\triangleq \left\{\mathbf{x}_{ij}\ \middle|\ 
% x_{\text{near}}\le x_{ij}\le x_{\text{far}},\ 
% |\phi_{ij}|\le \frac{\Phi}{2},\ 
% |\xi_{ij}|\le \frac{\Theta}{2}
% \right\}.
% \end{equation}
\begin{equation}\label{eq:frustum-set}
\mathcal{S}_{ij}^\mathbf{q}
\triangleq
\left\{
\mathbf{q}_{ij} \in \mathbb{R}^3\ \middle|\ 
\begin{aligned}
&x_\mathrm{near} \le x_{ij} \le x_\mathrm{far}\\
&|y_{ij}| \le x_{ij}\tan(\Phi/2)\\
&|z_{ij}| \le x_{ij}\tan(\Theta/2)
\end{aligned}
\right\}.
\end{equation}

\begin{problem}\label{prob}
Given a desired formation $\mathbf{x}^d_{ij}(t) \triangleq 
\begin{bmatrix}
L^d_{ij}(t) & \phi^d_{ij}(t) & \xi^d_{ij}(t) & \alpha^d_{ij}(t)
\end{bmatrix}^{\top}$ for the leader--follower pair, let $\mathbf{e}_{ij}\triangleq \mathbf{x}_{ij}-\mathbf{x}^d_{ij}$ denote the tracking error. The problem considered in this paper is to design perception-aware distributed 3D-LFF control laws for the followers such that:
\begin{enumerate}
    \item \textbf{Formation tracking:} when the perception constraints are inactive, $\mathbf{e}_{ij}(t)$ converges exponentially to $\mathbf{0}$.
    \item \textbf{Perception safety:} the leader remains visible, i.e., $\mathbf{q}_{ij}(t)\in\mathcal{S}_{ij}^\mathbf{q}$ for all $t\ge 0$ whenever $\mathbf{q}_{ij}(0)\in\mathcal{S}_{ij}^\mathbf{q}$.
\end{enumerate}
\end{problem}

% To address Problem~\ref{prob}, we develop a nominal distributed 3D-LFF controller and a CBF-QP safety filter that enforces the frustum constraints in real time.

\begin{figure*}[!t]
 \centering
 \includegraphics[width=\textwidth]{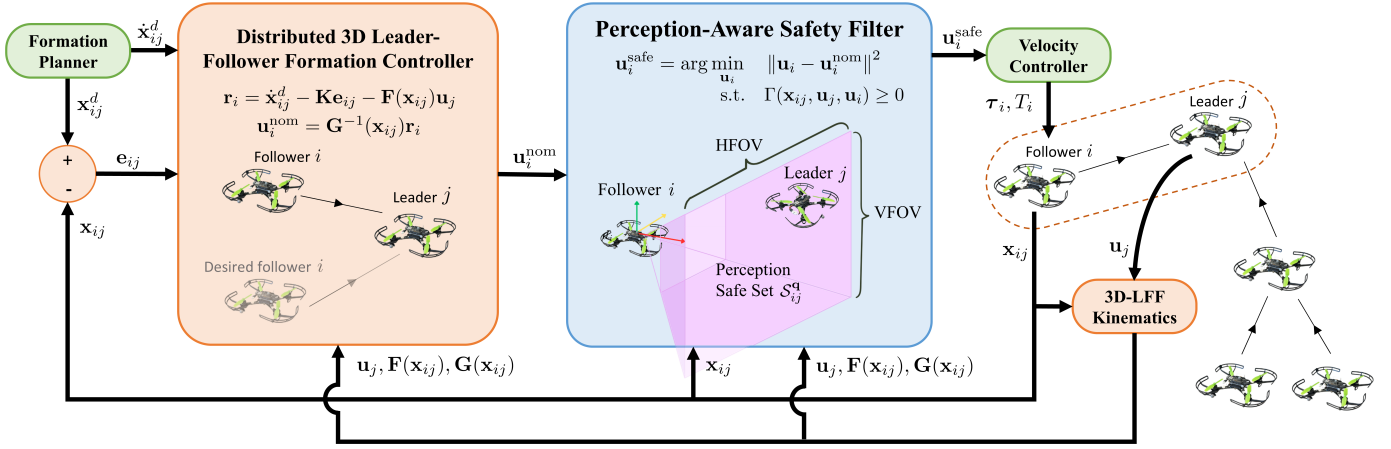}
\caption{Perception-aware safe 3D-LFF control architecture.
A high-level distributed 3D-LFF controller tracks the desired state $\mathbf{x}^d_{ij}(t)$ to generate the nominal command $\mathbf{u}_i^{\mathrm{nom}}$.
A perception-aware safety filter modifies $\mathbf{u}_i^{\mathrm{nom}}$ to enforce visibility constraints $\Gamma_\cidx$ via a CBF-QP, outputting the safe command $\mathbf{u}_i^{\mathrm{safe}}$ to keep the leader inside the perception safe set $\mathcal{S}_{ij}^\mathbf{q}$.
Finally, a low-level velocity controller tracks $\mathbf{u}_i^{\mathrm{safe}}$ for the full quadrotor dynamics.}
 \label{fig:method_overview}
\end{figure*}
%% ----------------------

\section{Main Results}\label{sec:method}
The proposed perception-aware safe 3D-LFF control framework employs a hierarchical architecture: an outer loop generates safe velocity commands, which are subsequently tracked by an inner-loop velocity controller handling the full-order system dynamics. 
As illustrated in Fig.~\ref{fig:method_overview}, this outer loop utilizes a cascaded two-stage design. 
First, a nominal formation controller is designed based on the relative 3D-LFF kinematics to track the desired formation, generating a nominal velocity command.
To guarantee perception safety, this nominal velocity command is subsequently processed through a CBF-QP safety filter. 
This filter computes a minimally modified safe velocity command that ensures the forward invariance of the perception safe set, while deviating minimally from the tracking objective. 
% The following subsections systematically derive and explain each component of the proposed system.

\subsection{Distributed 3D-LFF Controller}
% To track the velocity commands while accounting for the nonlinear full-order dynamics of the UAVs, the framework employs a hierarchical architecture. 
% In this setup, the high-level control interface, consisting of body-horizontal velocities and yaw rates, is regulated by a low-level velocity controller.
% We define this control interface as

The velocity command for UAV $k \in \{i, j\}$ consists of velocities and yaw rates expressed in $\mathcal{F}_{\psi_k}$, defined as
\begin{equation}\nonumber
% \label{eq:vbar_cmd2}
\mathbf{u}^{\mathrm{cmd}}_k \triangleq
\begin{bmatrix}
\overline{v}_{x,k}^{\mathrm{cmd}} & \overline{v}_{y,k}^{\mathrm{cmd}} & \overline{v}_{z,k}^{\mathrm{cmd}} & \overline{\omega}_{z,k}^{\mathrm{cmd}}
\end{bmatrix}^\top.
\end{equation}
Accordingly, let $\mathbf{u}_k \triangleq \begin{bmatrix} \overline{v}_{x,k} & \overline{v}_{y,k} & \overline{v}_{z,k} & \overline{\omega}_{z,k} \end{bmatrix}^\top$ denote the actual velocity and yaw rate of the quadrotor.
In this distributed architecture, the $i$-th follower computes its control using 3D-LFF states and the leader's communicated $\mathbf{u}_j$. 
Before designing this control action, we first formulate the 3D-LFF kinematics that relate these velocities to the formation evolution.

\begin{lemma}\label{lem:GF_matrices}
Consider 3D-LFF states $\mathbf{x}_{ij}=\begin{bmatrix}L_{ij} & \phi_{ij} & \xi_{ij} & \alpha_{ij}\end{bmatrix}^{\top}$ as shown in Fig.~\ref{fig:leader--follower-geometry}. The relative 3D-LFF kinematics can be written in the control-affine form
\begin{equation}
\dot{\mathbf{x}}_{ij}
=
\mathbf{F}(\mathbf{x}_{ij}) \mathbf{u}_j + \mathbf{G}(\mathbf{x}_{ij}) \mathbf{u}_i,
\label{eq:rel_kin}
\end{equation}
where the matrix-valued functions $\mathbf{F},\mathbf{G}:\mathbb{R}^4\rightarrow\mathbb{R}^{4\times4}$ are given by
\begin{align}\label{eq:FG}
\mathbf{F}(\mathbf{x}_{ij}) &=
\begin{bmatrix}
c_{\xi_{ij}} c_{\alpha_{ij}} & - c_{\xi_{ij}} s_{\alpha_{ij}} & s_{\xi_{ij}} & 0 \\[4pt]
\dfrac{s_{\alpha_{ij}}}{L_{ij} c_{\xi_{ij}}} & \dfrac{c_{\alpha_{ij}}}{L_{ij} c_{\xi_{ij}}} & 0 & 0 \\[8pt]
- \dfrac{s_{\xi_{ij}} c_{\alpha_{ij}}}{L_{ij}} & \dfrac{s_{\xi_{ij}} s_{\alpha_{ij}}}{L_{ij}} & \dfrac{c_{\xi_{ij}}}{L_{ij}} & 0\\[8pt]
- \dfrac{s_{\alpha_{ij}}}{L_{ij} c_{\xi_{ij}}} & - \dfrac{c_{\alpha_{ij}}}{L_{ij} c_{\xi_{ij}}} & 0 & 1 
\end{bmatrix},
\\
% \end{align}
% \begin{align}
\mathbf{G}(\mathbf{x}_{ij}) &=
\begin{bmatrix}
- c_{\xi_{ij}} c_{\phi_{ij}} & - c_{\xi_{ij}} s_{\phi_{ij}} & - s_{\xi_{ij}} & - d_i c_{\xi_{ij}} s_{\phi_{ij}} \\[4pt]
\dfrac{s_{\phi_{ij}}}{L_{ij} c_{\xi_{ij}}} & - \dfrac{c_{\phi_{ij}}}{L_{ij} c_{\xi_{ij}}} & 0 & - \dfrac{d_i c_{\phi_{ij}}}{L_{ij} c_{\xi_{ij}}} - 1 \\[8pt]
\dfrac{s_{\xi_{ij}} c_{\phi_{ij}}}{L_{ij}} & \dfrac{s_{\xi_{ij}} s_{\phi_{ij}}}{L_{ij}} & - \dfrac{c_{\xi_{ij}}}{L_{ij}} & \dfrac{d_i s_{\xi_{ij}} s_{\phi_{ij}}}{L_{ij}} \\[8pt]
- \dfrac{s_{\phi_{ij}}}{L_{ij} c_{\xi_{ij}}} & \dfrac{c_{\phi_{ij}}}{L_{ij} c_{\xi_{ij}}} & 0 & \dfrac{d_i c_{\phi_{ij}}}{L_{ij} c_{\xi_{ij}}}
\end{bmatrix},\nonumber
\end{align}
where $c_{(\cdot)}$ and $s_{(\cdot)}$ denote $\cos(\cdot)$ and $\sin(\cdot)$, respectively.
On the domain $\mathcal{S}_{ij}^\mathbf{q}$ where $L_{ij} \in \mathbb{R}_{>0}$ and $c_{\xi_{ij}} > 0$, the functions $\mathbf{F}$ and $\mathbf{G}$ are locally Lipschitz continuous. Furthermore, the input matrix $\mathbf{G}$ remains non-singular throughout this domain, ensuring that the subsequent formation control law and safety filter are well-defined.
% As will be shown in Section \ref{sec:perception_safety}, this condition is strictly enforced by the perception safe set $\mathcal{S}_{ij}^\mathbf{q}$.}
% To facilitate the algebraic cancellation of nonlinearities in the subsequent state feedback linearization (SFL) design, the analytical inverse of $\mathbf{G}(\mathbf{x}_{ij})$ is given by
% \begin{equation}
% \setlength{\arraycolsep}{2.5pt} % default is ~5pt
% \renewcommand{\arraystretch}{1.1} % row spacing
% \mathbf{G}^{-1}(\mathbf{x}_{ij})=
% \begin{bmatrix}
% - c_{\phi_{ij}} c_{\xi_{ij}} & 0 & L_{ij}s_{\xi_{ij}} c_{\phi_{ij}} & - L_{ij} s_{\phi_{ij}} c_{\xi_{ij}} \\
% - s_{\phi_{ij}} c_{\xi_{ij}} & d_i & L_{ij} s_{\xi_{ij}} s_{\phi_{ij}} & L_{ij} c_{\phi_{ij}} c_{\xi_{ij}} + d_i \\
% - s_{\xi_{ij}} & 0 & - L_{ij} c_{\xi_{ij}} & 0 \\
% 0 & -1 & 0 & -1
% \end{bmatrix}.\nonumber
% % \label{eq:G_inverse}
% \end{equation}
\end{lemma}
\begin{proof}
See Appendix~\ref{app:los_derivation}.
\end{proof}

\begin{theorem}\label{thm:sfl_tracking}
Consider the 3D-LFF kinematics \eqref{eq:rel_kin} and the formation tracking error $\mathbf{e}_{ij}\triangleq \mathbf{x}_{ij} - \mathbf{x}^d_{ij}$. Assume that the low-level velocity tracking is ideal such that $ \mathbf{u}_i^{\mathrm{nom}} = \mathbf{u}_i$.
Define the nominal formation tracking control law as
\begin{equation}\label{eq:sfl_control}
\mathbf{u}_i^{\mathrm{nom}} \triangleq \mathbf{G}^{-1}(\mathbf{x}_{ij})\,
\big(\dot{\mathbf{x}}^d_{ij} - \mathbf{K} \mathbf{e}_{ij} - \mathbf{F}(\mathbf{x}_{ij}) \mathbf{u}_j\big),
\end{equation}
where $\mathbf{K}\succ 0$ is a diagonal gain matrix. 
Assuming the relative state remains within the perception safe set ($\mathbf{q}{ij}(t) \in \mathcal{S}{ij}$ for all $t \ge 0$), the formation tracking error $\mathbf{e}_{ij}$ converges exponentially to $\mathbf{0}$.
\end{theorem}

\begin{proof}
Consider the 3D-LFF kinematics from Lemma~\ref{lem:GF_matrices} and define the formation tracking error as $\mathbf{e}_{ij} \triangleq \mathbf{x}_{ij} - \mathbf{x}^d_{ij}$. The resulting error dynamics are given by $\dot{\mathbf{e}}_{ij} = \mathbf{F}(\mathbf{x}_{ij})\mathbf{u}_j + \mathbf{G}(\mathbf{x}_{ij})\mathbf{u}_i - \dot{\mathbf{x}}^d_{ij}$. To shape this response, we design the nominal control law $\mathbf{u}_i^{\mathrm{nom}} \triangleq \mathbf{G}^{-1}(\mathbf{x}_{ij}) \mathbf{r}_i$, where the virtual input is defined as $\mathbf{r}_{i} \triangleq \dot{\mathbf{x}}^d_{ij} - \mathbf{K} \mathbf{e}_{ij} - \mathbf{F}(\mathbf{x}_{ij}) \mathbf{u}_j$. Per Lemma \ref{lem:GF_matrices}, $\mathbf{G}$ is non-singular on $\mathcal{S}_{ij}^\mathbf{q}$, ensuring $\mathbf{u}_i^{\mathrm{nom}}$ is well-defined. Substituting $\mathbf{u}_i^{\mathrm{nom}}$ into the error dynamics cancels the leader-drift and feedforward terms, yielding the closed-loop system:
\begin{equation}\label{eq:err_dyn_final}
\dot{\mathbf{e}}_{ij} = \mathbf{F}\mathbf{u}_j + \mathbf{G}(\mathbf{G}^{-1}\mathbf{r}_i) - \dot{\mathbf{x}}^d_{ij} = -\mathbf{K}\mathbf{e}_{ij}.
\end{equation}
Given $\mathbf{K} \succ 0$, the linear system \eqref{eq:err_dyn_final} is exponentially stable. Specifically, the error satisfies $\|\mathbf{e}_{ij}(t)\| \le e^{-\lambda_{\min}(\mathbf{K})t} \|\mathbf{e}_{ij}(0)\|$ for all $t \ge 0$, where $\lambda_{\min}(\mathbf{K})$ denotes the smallest eigenvalue of $\mathbf{K}$. This confirms that $\mathbf{e}_{ij}(t)$ converges exponentially to $\mathbf{0}$ at a rate of $\lambda_{\min}(\mathbf{K})$.
\end{proof}

\subsection{Perception-Aware Safety Filter}\label{sec:perception_safety}
While the nominal formation control law~\eqref{eq:sfl_control} successfully achieves formation tracking, it is fundamentally agnostic to perception constraints. 
To guarantee reliable relative state estimation in vision-based LFF, the control framework must ensure that the leader continuously remains within the follower's perception-safe region in \eqref{eq:frustum-set}.
% We encode this requirement geometrically using the follower's camera frustum in Cartesian coordinates. 
% Specifically, we define the perception safe set as
% \begin{equation}\label{eq:frustum-set}
% \mathcal{S}_{ij}^\mathbf{q}^\mathbf{q}
% \triangleq
% \left\{
% \mathbf{q}_{ij}\in\mathbb{R}^3 \ \middle|\ 
% \begin{aligned}
% &x_\mathrm{near} \le x_{ij} \le x_\mathrm{far}\\
% &|y_{ij}| \le x_{ij}\tan(\Phi/2)\\
% &|z_{ij}| \le x_{ij}\tan(\Theta/2)
% \end{aligned}
% \right\}.
% \end{equation}
% where $\Phi,\Theta\in(0,\pi)$ denote the maximum horizontal and vertical FOV angles of the camera, and $x_\mathrm{near}, x_\mathrm{far}\in\mathbb{R}_{>0}$ specify the sensing range limits of the depth camera, respectively, satisfying $x_\mathrm{near}<x_\mathrm{far}$ (see Fig.~\ref{fig:leader--follower-geometry}).
% To guarantee that the leader remains strictly inside the follower's camera perception safe set, we augment the system with a CBF-QP safety filter. 
Recall that although the 3D-LFF state $\mathbf{x}_{ij}$ is expressed in spherical coordinates for compact control synthesis, the perception safe set $\mathcal{S}_{ij}^\mathbf{q}$ in~\eqref{eq:frustum-set} is naturally defined using the Cartesian relative position $\mathbf{q}_{ij} = [x_{ij}, y_{ij}, z_{ij}]^\top$. 
The geometric transformation from the state $\mathbf{x}_{ij}$ to $\mathbf{q}_{ij}$ is given by
\begin{equation}
\mathbf{q}_{ij}=\mathbf{q}(\mathbf{x}_{ij})=
\begin{bmatrix}
L_{ij}c_{\xi_{ij}}c_{\phi_{ij}} \\
L_{ij}c_{\xi_{ij}}s_{\phi_{ij}} \\
L_{ij}s_{\xi_{ij}}
\end{bmatrix}.
\label{eq:q_spherical}
\end{equation}
Formulating the safety filter in the Cartesian domain is highly advantageous because it allows the camera frustum boundaries to be expressed as linear inequalities. 
Accordingly, each boundary is encoded by a continuously differentiable function $h_\cidx:\mathbb{R}^3\to\mathbb{R}$ for $\cidx \in \{1, \dots, 6\}$:
\begin{equation}\label{eq:control-barrier-function}
\begin{aligned}
h_{1} &= x_{ij} - x_\mathrm{near}, 
&\; h_{2} &= x_\mathrm{far} - x_{ij},\\
h_{3} &= x_{ij} \tan(\Phi/2) + y_{ij}, 
&\; h_{4} &= x_{ij} \tan(\Phi/2) - y_{ij},\\
h_{5} &= x_{ij} \tan(\Theta/2) + z_{ij}, 
&\; h_{6} &= x_{ij} \tan(\Theta/2) - z_{ij},
\end{aligned}
\end{equation}
where each $h_\cidx$ corresponds to a specific geometric boundary of the frustum.
% \begin{remark}
% The CBFs $h_\cidx(\mathbf{q}(\mathbf{x}_{ij}))$ have a relative degree one with respect to the 3D-LFF kinematics in \eqref{eq:rel_kin}. 
% This implies that the control input $\mathbf{u}_i$ appears directly in the first-order time derivative $\dot{h}_\cidx$, facilitating the design of a first-order CBF-QP safety filter.
% \end{remark}

% By construction, the perception safe set are chosen such that $\mathcal{S}_{ij} \subset \mathcal{S}^q_{ij}$, ensuring that the perception-safe region lies entirely within the domain where the relative kinematics are well-defined (see Remark \ref{rem:FGLipschitz}).
Accordingly, we define the state-space perception safe set $\mathcal{S}_{ij}$ as the set of all states whose Cartesian mappings reside within $\mathcal{S}_{ij}^{\mathbf{q}}$, expressed as the intersection of 0-superlevel sets:
\begin{align}
\mathcal{S}_{ij} = \left\{ \mathbf{x}_{ij}  \in \mathbb{R}^4\mid h_\cidx(\mathbf{q}(\mathbf{x}_{ij})) \ge 0, \ \forall \cidx \in \{1, \dots, 6\} \right\}.
\end{align}
To enforce forward invariance of the safe set $\mathcal{S}_{ij}$ for the system \eqref{eq:rel_kin}, we impose the following CBF constraints
\begin{align}
\Gamma_\cidx(\mathbf{x}_{ij},\mathbf{u}_j,\mathbf{u}_i)
\triangleq&
\dot{h}_\cidx(\mathbf{x}_{ij},\mathbf{u}_j,\mathbf{u}_i)
+\kappa h_\cidx(\mathbf{q}(\mathbf{x}_{ij}))\nonumber\\
\ge&0,\quad\forall \cidx\in\{1, \dots, 6\},
\label{eq:Gamma_def}
\end{align}
where $\kappa\in \mathbb{R}_{>0}$ is a user-defined constant gain.
At each control timestep, the safe command is obtained by solving the following QP
\begin{align}\label{eq:CBF-QP}
\mathbf{u}^\mathrm{safe}_i
=
\arg\min_{\mathbf{u}_i} \quad
 &
\left\|
\mathbf{u}_i - \mathbf{u}_i^{\mathrm{nom}}
\right\|^2\\
\text{s.t.} \quad &
\Gamma_\cidx(\mathbf{x}_{ij},\mathbf{u}_j,\mathbf{u}_i)\ge 0,\quad\forall \cidx \in \{1, \dots, 6\}. \nonumber
\end{align}

\begin{theorem}\label{thm:cbf_qp_safety}
Consider the 3D-LFF kinematics $\dot{\mathbf{x}}_{ij} = \mathbf{F}(\mathbf{x}_{ij}) \mathbf{u}_j + \mathbf{G}(\mathbf{x}_{ij}) \mathbf{u}_i$ and the CBFs $h_\cidx(\mathbf{q}_{ij})$ in~\eqref{eq:control-barrier-function}. 
Let $\mathcal{S}_{ij}$ denote the perception safe set in the 3D-LFF state space. 
Suppose that, for all $t\ge 0$, the safety filter computes an input $\mathbf{u}_i^{\mathrm{safe}}(t)$ that satisfies the CBF constraints~\eqref{eq:Gamma_def} and the low-level velocity tracking
is ideal such that $\mathbf{u}^\mathrm{safe}_i = \mathbf{u}_i$.
Then, for any initial condition $\mathbf{x}_{ij}(0)\in\mathcal{S}_{ij}$, the set $\mathcal{S}_{ij}$ is forward invariant under~\eqref{eq:rel_kin}. 
\end{theorem}

\begin{proof}
By the chain rule, the time derivative of the CBFs can be written as
\begin{align*}
    \dot{h}_\cidx(\mathbf{x}_{ij},\mathbf{u}_j,\mathbf{u}_i) =\nabla_{\mathbf{q}_{ij}} h_\cidx^\top  \frac{\partial \mathbf{q}_{ij}}{\partial \mathbf{x}_{ij}}  \dot{\mathbf{x}}_{ij}.
\end{align*}
% where $\nabla_{\mathbf{q}_{ij}} h_\cidx\in \mathbb{R}^{3}$ is is the gradient vector of the $k$-th control barrier function with respect to $\mathbf{q}_{ij}$ and $\mathbf{J}_{\mathbf{x}_{ij}} \mathbf{q}\in \mathbb{R}^{3 \times 4}$ is the Jacobian matrix of the coordinate transformation $\mathbf{q}(\mathbf{x}_{ij})$ with respect to $\mathbf{x}_{ij}$.
Since each $h_\cidx(\mathbf{q}_{ij})$ in~\eqref{eq:control-barrier-function} is affine in $\mathbf{q}_{ij}$, its gradient $\nabla_{\mathbf{q}_{ij}} h_\cidx = [\frac{\partial h_\cidx}{\partial x_{ij}}, \frac{\partial h_\cidx}{\partial y_{ij}}, \frac{\partial h_\cidx}{\partial z_{ij}}]^\top$  is constant.
Moreover, the mapping $\mathbf{q}(\mathbf{x}_{ij})$ in~\eqref{eq:q_spherical} is continuously differentiable on the domain $\mathcal{S}_{ij}$.
% , hence, $\mathbf{J}_{\mathbf{x}_{ij}}\mathbf{q}(\mathbf{x}_{ij})$ is locally Lipschitz on $\mathcal{S}_{ij}$.
By construction, the CBF-QP safety filter selects a control input $\mathbf{u}_i^{\mathrm{safe}}(t)$ that is locally Lipschitz (see \cite{ames2016control}) and satisfies the CBF constraints $\dot h_\cidx(t)\;\ge\; -\kappa\, h_\cidx(t)$.
According to Lemma \ref{lem:GF_matrices}, $\mathbf{F}$ and $\mathbf{G}$ are locally Lipschitz on the domain $\mathcal{S}_{ij}$.
This ensures that the system dynamics $\dot{\mathbf{x}}_{ij}$ and, consequently, $\dot{h}_\cidx$ are locally Lipschitz.
By the Comparison Lemma, it follows that:
\[
h_\cidx(t)\;\ge\; e^{-\kappa t}\,h_\cidx(0),\quad \forall t\ge 0.
\]
Therefore, if $\mathbf{x}_{ij}(0)\in\mathcal{S}_{ij}$, then $h_\cidx(0) \ge 0$, which implies $h_\cidx(t) \ge 0$ for all $t \ge 0$ and all $\cidx \in\{1,\ldots,6\}$. 
\end{proof}

% [Implementation Robustness and Feasibility]
\begin{remark}\label{rem:buffer_slack}
To ensure the feasibility of the QP in \eqref{eq:CBF-QP} under strict control input bounds $\mathbf{u}_i \in \mathcal{U}$ or aggressive leader maneuvers, we relax the CBF constraints using non-negative slack variables $\zeta_\cidx \in \mathbb{R}_{\ge0}$. The modified constraint is given by $\dot{h}_\cidx + \kappa h_\cidx \ge \delta_\cidx - \zeta_\cidx$, where $\zeta_\cidx$ is heavily penalized in the objective function to prioritize safety. The robustness margin $\delta_\cidx\in \mathbb{R}_{\ge0}$ provides a conservative buffer that accounts for unmodeled low-level tracking errors and external disturbances, ensuring the leader remains well within the FOV even under non-ideal conditions (see \cite{compton2024learning}).
\end{remark}

\section{Experiments}\label{sec:experiments}

This section evaluates the proposed perception-aware safe 3D-LFF controller on two platforms:
(A) a simulated platform in Gazebo utilizing the CrazySim framework \cite{LlanesICRA2024}, and (B) a physical hardware platform using the Crazyflie 2.1 brushless quadrotor.
For a direct comparative analysis, the multi-robot system comprises one leader and two followers, both executing the proposed method: one without the safety filter (NoCBF), and the other with the safety filter (CBF).
% ========= ARXIV =================
The approach is evaluated under two scenarios: (i) a three-stage task, and (ii) an abrupt leader motion task. In Scenario (i), to analyze the effectiveness of the safety filter, the evaluation task is divided into three stages. Stage 1 represents the nominal condition where no safety conflict exists, Stage 2 introduces a safety-conflicting formation, and Stage 3 removes the safety conflict. Scenario (ii) focuses on the transient safety response to aggressive leader motion, specifically a sudden leader deceleration.
% ========= IEEE-RAL =================
% To analyze the effectiveness of the safety filter, the evaluation task is divided into three stages. Stage 1 represents the nominal condition where no safety conflict exists, Stage 2 introduces a safety-conflicting formation, and Stage 3 removes the safety conflict. 

Due to payload and sensing constraints on the Crazyflie, onboard relative sensing is not used in experiments.
Instead, a motion-capture system provides global poses, from which the 3D-LFF states are computed and used as inputs to the proposed distributed controller and safety filter. 
% A host pc sends the velocity command at 20Hz. 
% The velocity command is executed by the onboard geometric controller running at 100Hz. 
This setup emulates ideal relative measurements and isolates the contribution of the control and safety filter layers. 
Simulation and experimental results demonstrate that the proposed framework successfully tracks desired formations while strictly prioritizing perception safety, whereas the NoCBF baseline fails to maintain the leader within the perception safe set.

\subsection{Gazebo Simulation}\label{sec:gazebo_exp}

The proposed framework is first evaluated in a confined cave environment, simulating a GPS-denied setting that requires strict vision-based coordination.
As visualized in Fig.~\ref{fig:scenario}, the terrain is designed to naturally trigger the proposed test scenario. 
% ========= ARXIV =================
An open chamber within the cave allows the multi-UAV team to expand their nominal formation for exploration tasks, introducing the safety-conflicting formation evaluated in Stage 2 of Scenario (i). Subsequently, navigating through closed corridors forces the leader to execute sharp, sudden maneuvers, providing the conditions for the abrupt motion test in Scenario (ii).
% ========= IEEE-RAL =================
% An open chamber within the cave allows the multi-UAV team to expand their nominal formation for exploration tasks, introducing the safety-conflicting formation evaluated in Stage 2.

\begin{figure}[!t]
 \centering
% ========= ARXIV =================
 \includegraphics[width=1\linewidth]{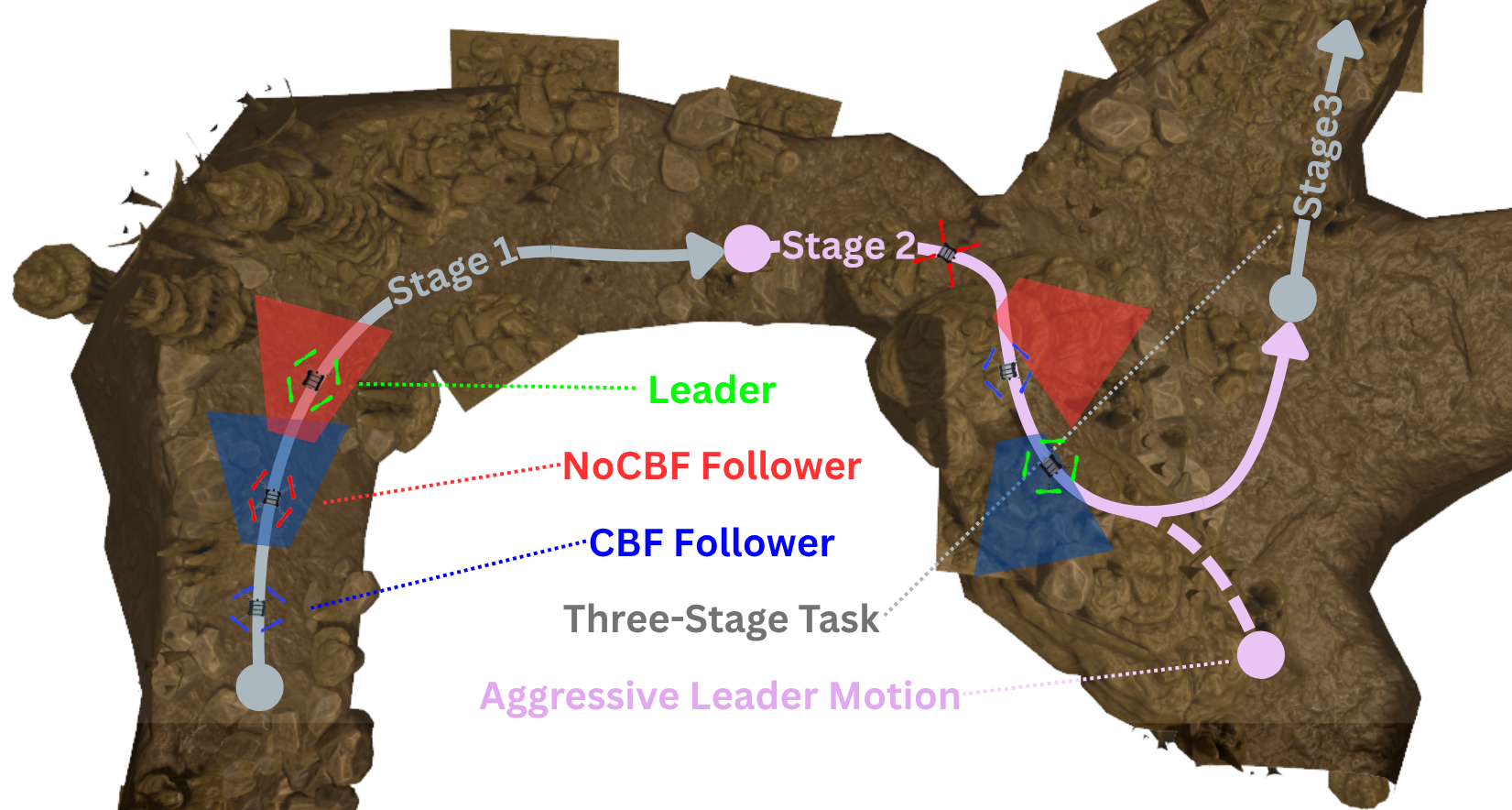}
\caption{Multi-UAV leader--follower navigation in a Gazebo cave environment. The leader (green), CBF-equipped follower (blue), and NoCBF baseline (red) are shown with their respective camera frustums (shaded regions). The evaluation paths include a three-stage task introducing a safety-conflicting formation, followed by an abrupt leader motion task to thoroughly evaluate the proposed method.}
% ========= IEEE-RAL =================
 % \includegraphics[width=1\linewidth]{figures/scenario.png}
 % \caption{Multi-UAV leader--follower navigation in a Gazebo cave environment. The leader (green), CBF-equipped follower (blue), and NoCBF baseline (red) are shown with their respective camera frustums (shaded regions). The evaluation paths include a three-stage task introducing a safety-conflicting formation to evaluate the proposed method.}
 \label{fig:scenario}
\end{figure}

% ========= ARXIV =================
\subsubsection{Three-Stage Task}\label{sec:three}
% ========= ARXIV =================
\begin{figure}[!t]
 \centering
 \includegraphics[width=1\linewidth]{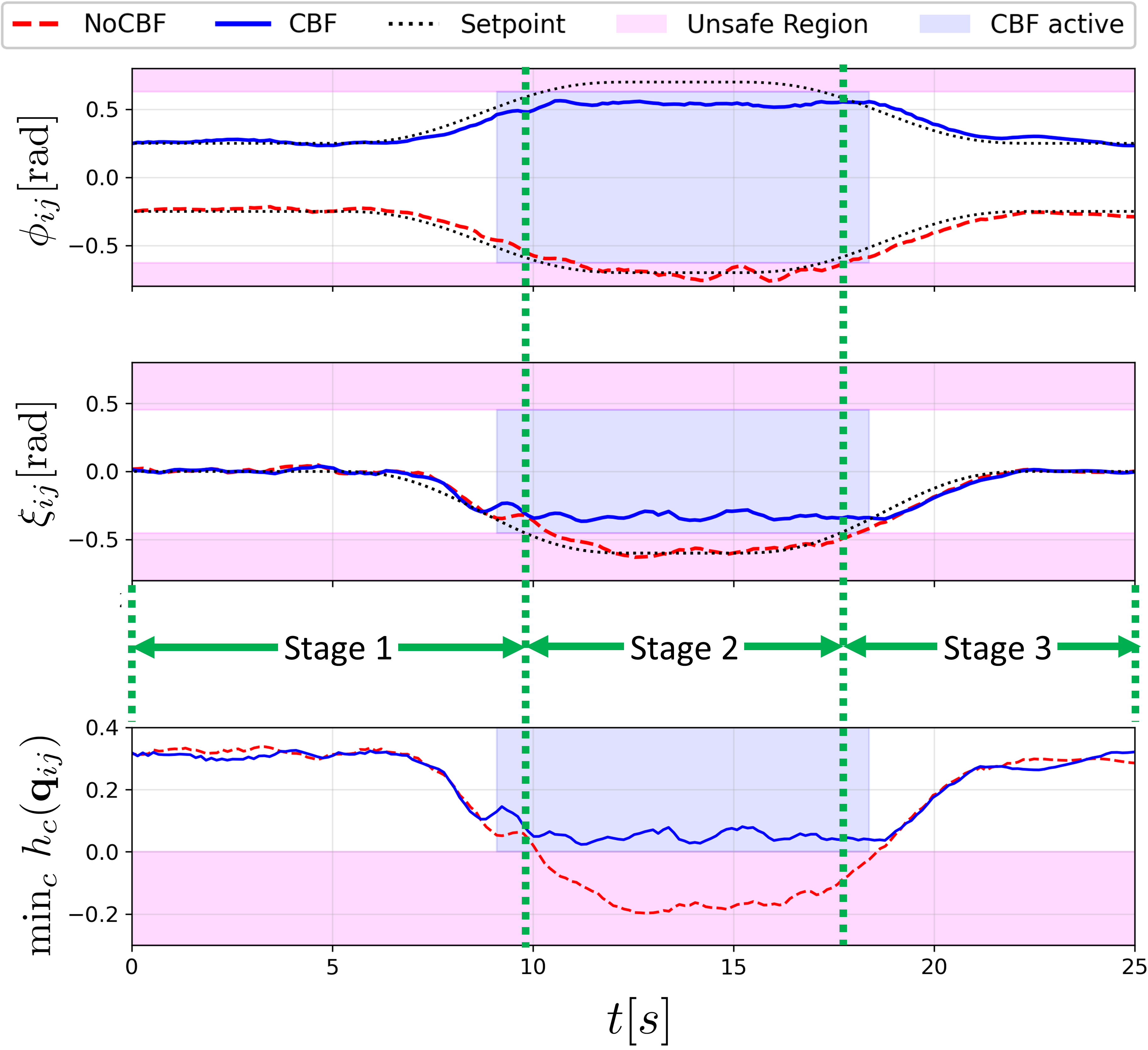}
 \caption{Gazebo simulation of the three-stage task. A comparison of 3D-LFF states (top and middle) and minimum CBF value (bottom). During Stage 2, desired setpoints approach the unsafe region (shaded in magenta). Consequently, the NoCBF follower (dashed red line) violates the safety constraints, whereas the CBF follower (solid blue line) remains strictly safe ($\min_\cidx\, h_\cidx(\mathbf{q}_{ij}) \ge 0$). Blue shaded regions indicate active safety filter intervention.}
 \label{fig:crazysim_conflicting}
\end{figure}
The proposed scheme is evaluated and compared against the baseline under a three-stage task scenario, with the corresponding Gazebo simulation results depicted in Fig.~\ref{fig:crazysim_conflicting}. 
The analysis focuses on the trajectories of the horizontal angle $\phi_{ij}$ and elevation angle $\xi_{ij}$, as these angular states are directly constrained by the camera frustum. 
Furthermore, the bottom plot illustrates the minimum barrier value ($\min_\cidx\, h_\cidx(\mathbf{q}_{ij})$), which serves as a unified metric for perception safety.

In the first stage, the team navigates a narrow tunnel where each follower is able to track its nominal setpoints, successfully maintaining a V-formation.
In the second stage, the formation is expanded to increase environmental coverage. 
This formation change introduces a conflict between formation and perception safety objectives. 
To resolve this, the proposed CBF-QP framework dynamically prioritizes safety over nominal tracking. As the leader approaches the boundary of the camera frustum ($9.09\text{ s} < t < 18.38\text{ s}$), the active CBF filter intervenes, ensuring the leader remains strictly within the perception safe set ($\min_\cidx\, h_\cidx(\mathbf{q}_{ij}) \ge 0$).
On the other hand, the NoCBF follower fails to adapt, resulting in severe safety violations as its states enter the unsafe regions and its minimum barrier value drops below zero ($\min_\cidx\, h_\cidx(\mathbf{q}_{ij}) < 0$).
In the third stage, the conflict is removed, and both followers are able to converge to their nominal formation objectives.

% ========= ARXIV =================
\subsubsection{Abrupt Leader Motion}\label{sec:abrupt}

\begin{figure}[!t]
\centering
\includegraphics[width=1\linewidth]{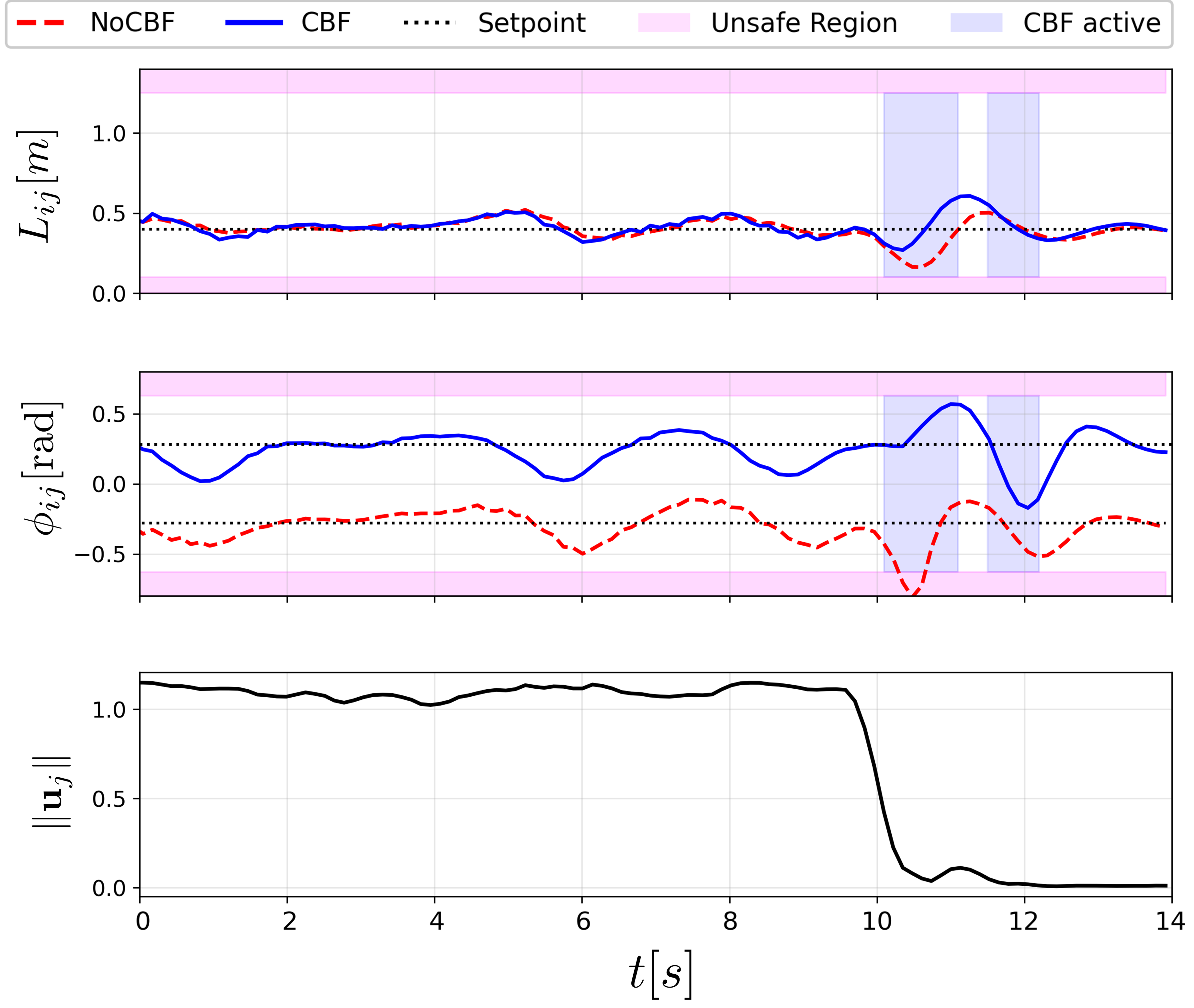}
\caption{Gazebo simulation of abrupt leader motion. A comparison of 3D-LFF states (top and middle) and Euclidean norm of leader velocity (bottom). Sudden leader deceleration rapidly decreases $L_{ij}$ and disturbs $\phi_{ij}$, driving the nominal formation toward unsafe region (shaded in magenta). Consequently, the NoCBF baseline (dashed red line) violates constraints, whereas the CBF controller (solid blue line) successfully maintain safety.}
\label{fig:crazysim_abrupt}
\end{figure}

The proposed scheme is evaluated and compared against the baseline under an abrupt leader motion, with the corresponding Gazebo simulation results depicted in Fig.~\ref{fig:crazysim_abrupt}. 
To analyze the transient response to the abrupt leader motion, the plots focus on the camera-to-leader distance $L_{ij}$ and the horizontal angle $\phi_{ij}$, respectively. 
These specific states are highlighted because a sudden leader deceleration rapidly compresses the formation and induces lateral drift, directly threatening the perception safety. 
Finally, the bottom plot displays the Euclidean norm of the leader's velocity $||\mathbf{u}_j||$, explicitly illustrating the magnitude and timing of the sudden deceleration maneuver.

In this scenario, the leader executes an initial maneuver with faster speed where both followers track a tight formation. 
Reaching a dead end, the leader suddenly decelerates, indicated by a sharp drop in $||\mathbf{u}_j||$. 
The baseline fails to handle the sharp transient, causing the follower to aggressively overshoot and violating the perception safety on $\phi_{ij}$. 
In contrast, the proposed approach successfully mitigates this by modifying the control actions to bound the overshoot.
The safety filter enforces the perception constraints, keeping the leader within the follower's camera frustum under the leader's abrupt movement. 
Once the leader's velocity settles to zero, the follower smoothly converges back to the desired tracking setpoints.
% ========= ARXIV =================

\subsection{Hardware Experiments}\label{sec:hw_exp}

The hardware experiments validate the proposed framework under real-world flight conditions, capturing the effects of unmodeled aerodynamic drag, communication latency, and actuator noise. 
The system is evaluated using the identical leader--follower protocol from the simulations, repeating the three-stage task to quantitatively assess perception safety and formation tracking during physical flight. 
Furthermore, to rigorously stress-test the proposed method during these maneuvers, the leader is commanded to track a yaw-modulated lemniscate trajectory, as depicted in Figs.~\ref{fig:exp_top_view} and \ref{fig:exp_3d_traj}.
\begin{figure}[!t]
 \centering
 % First Subfigure: Top View
 \subfloat[Top view of the experiment.\label{fig:exp_top_view}]{
 \includegraphics[width=\linewidth]{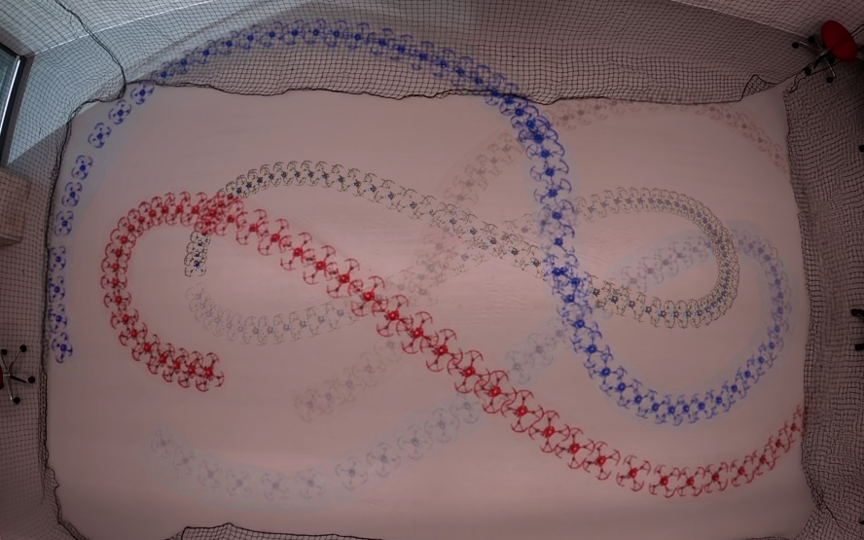}
 }

 % Second Subfigure: 3D Trajectory
 \subfloat[3D perspective of the trajectories.\label{fig:exp_3d_traj}]{
 \includegraphics[width=\linewidth]{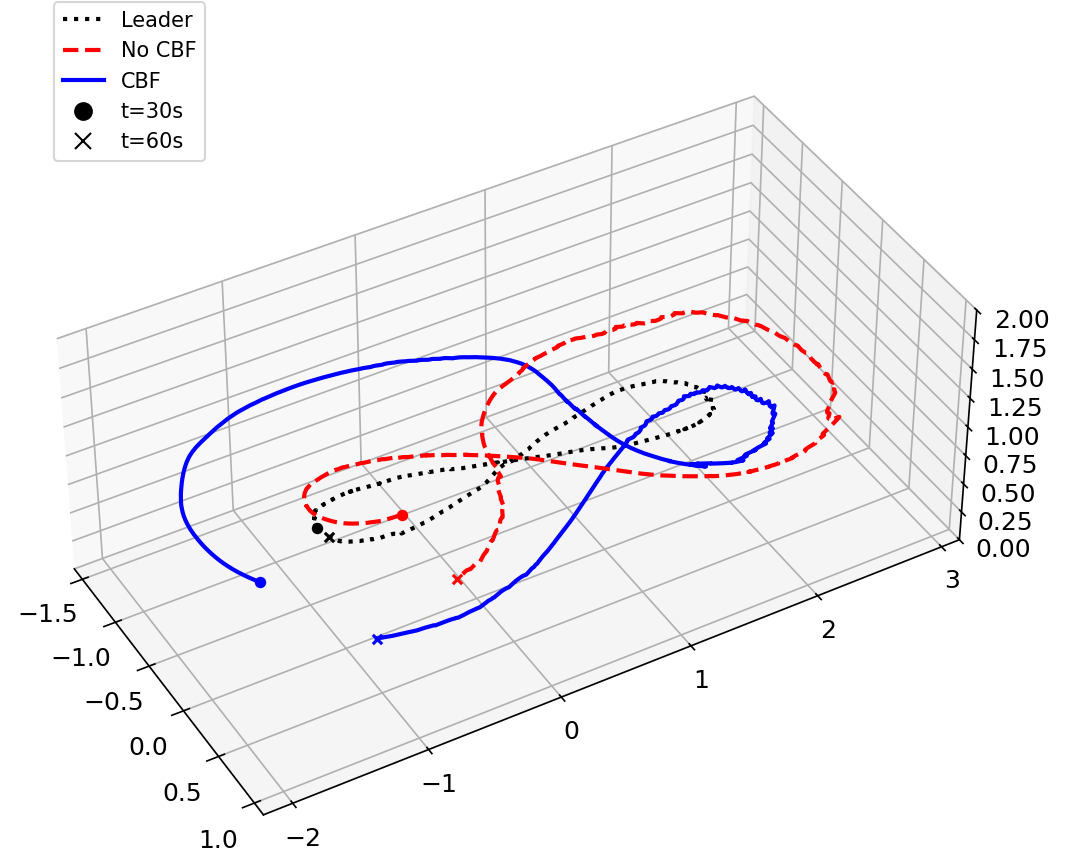}
 }
 
 \caption{Trajectories of the leader and the followers during the three-stage task scenario for the hardware experiment. For clarity, we show only the lemniscate trajectory from the second half of the experiment. Markers indicate the positions at 30s and 60s.}
 \label{fig:exp_combined_trajectory}
\end{figure}

\begin{figure}[!t]
 \centering
 \includegraphics[width=1\linewidth]{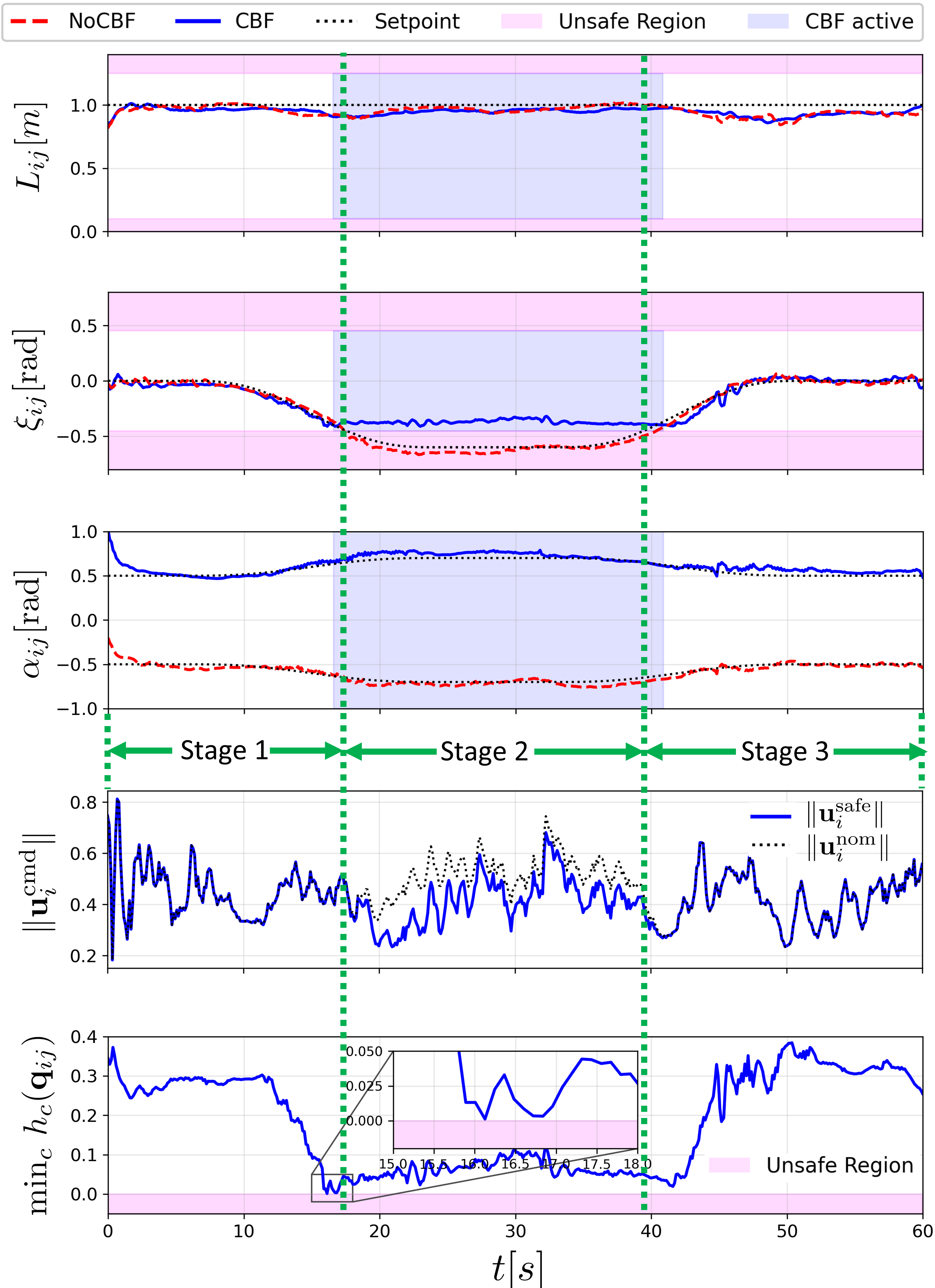}
 \caption{Hardware experiment of the three-stage task. A comparison of 3D-LFF states (first 3 plots) and control input metrics (last 2 plots). During the Stage 2 conflict, desired setpoints approach the unsafe region (shaded in magenta). Consequently, the NoCBF baseline (dashed red line) violates the safety constraints, whereas the CBF follower (solid blue line) remains strictly safe ($\min_\cidx\, h_\cidx(\mathbf{q}_{ij}) > 0$). Blue shaded regions indicate active safety filter intervention ($\|\mathbf{u}_i^{\mathrm{safe}}\| \neq \|\mathbf{u}_i^{\mathrm{nom}}\|$).}
 \label{fig:exp_shorten}
\end{figure}

The hardware experimental results for the three-stage task are presented in Fig.~\ref{fig:exp_shorten}. 
The analysis focuses on the trajectories of the distance $L_{ij}$, elevation angle $\xi_{ij}$, and relative heading $\alpha_{ij}$ to evaluate the framework's ability to seamlessly resolve perception safety conflicts and preserve formation tracking during periods of active CBF intervention. 
Furthermore, the fourth plot displays the comparison of $\|\mathbf{u}_i^{\mathrm{nom}}\|$ and $\|\mathbf{u}_i^{\mathrm{safe}}\|$ to illustrate the active velocity command adjustments made by the safety filter, while the last plot illustrates the minimum barrier value ($\min_\cidx\, h_\cidx(\mathbf{q}_{ij})$), which serves as a unified metric for perception safety.

During Stage 1, both controllers are able to track the desired formation, yielding time-averaged absolute errors of approximately $\SI{0.040}{m}$ for $L_{ij}$, $\SI{0.020}{\radian}$ for $\xi_{ij}$, and $\SI{0.058}{\radian}$ for $\alpha_{ij}$.
In Stage 2, formation expansion drives the desired elevation angle $\xi_{ij}$ into the unsafe region. Consequently, the NoCBF baseline violates the safety constraints. 
Conversely, the proposed framework safely resolves this conflict by computing a minimally modified safe control input $\|\mathbf{u}_i^{\mathrm{safe}}\| \neq \|\mathbf{u}_i^{\mathrm{nom}}\|$. 
During this active intervention ($16.66\text{ s} < t < 40.94\text{ s}$), the filter strictly bounds $\xi_{ij}$ within the perception safe region ($\min_\cidx\, h_\cidx(\mathbf{q}_{ij}) > 0$) while simultaneously sustaining accurate tracking of the unconstrained states $L_{ij}$ and $\alpha_{ij}$. 
Finally, in Stage 3, the safety conflict is removed, and both followers are able to converge to their formation objectives.

% ============================================================
\section{Conclusion}

This paper presented a perception-aware safe distributed leader--follower formation control framework for multi-UAV systems equipped with body-fixed cameras and limited visibility.
We derived a 3D leader--follower formation (3D-LFF) kinematic representation in spherical coordinates and designed a nominal formation controller to track time-varying formation references.
To guarantee continuous visibility, we incorporated a CBF-QP safety filter that encodes camera-frustum constraints-including FOV limits and depth range-in Cartesian coordinates and computes minimally modified safe velocity commands.
The resulting architecture resolves conflicts between aggressive formation objectives and perception constraints by prioritizing visibility while preserving stable tracking and smooth recovery.
% ========= ARXIV =================
The method was evaluated in Gazebo simulations and hardware experiments under a three-stage formation task and an abrupt leader motion scenario.
Across these tests, when the visibility constraints were inactive, the nominal controller achieved stable formation tracking and convergence to the desired 3D-LFF reference.
When the desired formation conflicted with safety constraints, or when the system was subjected to sharp transients from sudden leader deceleration, the CBF-QP safety filter dynamically modified the commanded velocities to enforce perception safety. 
This ensured that the leader remained strictly within the follower's camera frustum at all times, with the followers smoothly recovering the formation once the conflict or disturbance passed.
% ========= IEEE-RAL =================
% The method was evaluated in Gazebo simulations and hardware experiments under representative scenarios.
% Across these tests, when the visibility constraints were inactive, the nominal controller achieved stable formation tracking and convergence to the desired 3D-LFF reference.
% When the desired formation conflicted with safety constraints, the CBF-QP safety filter modified the commanded velocities to enforce perception safety, ensuring that the leader remained within the follower's camera frustum at all time.
While the current experimental setup isolates the control and safety layers using motion capture, future work will integrate onboard perception to close the loop using raw image-based measurements.
% Additional directions include extending the formulation to limited-FOV cameras with occlusions and clutter, incorporating actuator and communication constraints, and developing robust or predictive barrier constructions to improve performance under latency and model mismatch.

\bibliographystyle{ieeetr}
\bibliography{references}

\appendix
\section{3D-LFF Kinematics Derivation}\label{app:los_derivation}

Formulating the kinematics of the spherical states $\mathbf{\dot{x}}_{ij}=[\dot{L}_{ij}, \dot{\phi}_{ij}, \dot{\xi}_{ij}, \dot{\alpha}_{ij}]^\top$ relies on determining the Cartesian relative velocity.
We begin by taking the time derivative of the relative position $\mathbf{q}_{ij} = \mathbf{R}_i^{\top}(\mathbf{p}_j-\mathbf{p}_i) - d_i\mathbf{b}_1$, which yields:
\begin{equation}
\dot{\mathbf{q}}_{ij} = \dot{\mathbf{R}}_i^{\top}(\mathbf{p}_j-\mathbf{p}_i) + \mathbf{R}_i^{\top}(\dot{\mathbf{p}}_j-\dot{\mathbf{p}}_i).
\end{equation}
% We implement the low-level controller using the velocity commands $\mathbf{u}_i=[\overline{v}_{x,i},\,\overline{v}_{y,i},\,\overline{v}_{z,i},\,\overline{\omega}_{z,i}]^\top$, with roll and pitch regulated implicitly by the inner-loop velocity controller. 
% Therefore, the attitude required in the kinematic relations is captured by the heading (yaw) only, i.e.,
We assume the quadrotor operates in a near-hover regime where roll and pitch angles remain small. 
Consequently, the rotation matrix $\mathbf{R}_k$ from the body frame $\mathcal{F}_{B_k}$ to the inertial frame $\mathcal{F}_I$ is approximated by the rotation $\mathbf{\overline{R}}_k = \mathbf{R}_z(\psi_k)$ from the yaw-aligned frame $\mathcal{F}_{\psi_k}$ to $\mathcal{F}_I$.
% \begin{equation}
% \mathbf{R}_i=\mathbf{R}_z(\psi_i), \quad \mathbf{R}_j=\mathbf{R}_z(\psi_j),
% \end{equation}
Let $\dot{\psi}_k = \overline{\omega}_{z,k}$, so that
\begin{equation}\label{eq:app_rhat}
\dot{\overline{\mathbf{R}}}_i = \overline{\mathbf{R}}_i[\overline{\omega}_{z,i} \mathbf{b}_3]_\times, \quad \dot{\overline{\mathbf{R}}}_i^{\top} = -[\overline{\omega}_{z,i} \mathbf{b}_3]_\times \overline{\mathbf{R}}_i^{\top}.
\end{equation}
% where $\mathbf{b}_3$ is the third standard basis vector of $\mathbb{R}^3$.
Let $\mathbf{\overline{v}}_k = [\overline{v}_{x,k}, \overline{v}_{y,k}, \overline{v}_{z,k}]^\top$ denote the velocity expressed in $\mathcal{F}_{\psi_k}$. 
The relationship between the inertial velocity $\dot{\mathbf{p}}_k$ and the yaw-aligned rotation is given by:
\begin{equation}
\mathbf{\overline{v}}_i \triangleq \overline{\mathbf{R}}_i^{\top}\dot{\mathbf{p}}_i, \quad \mathbf{\overline{v}}_j \triangleq \overline{\mathbf{R}}_j^{\top}\dot{\mathbf{p}}_j.
\end{equation}
Consequently,
\begin{equation}
\overline{\mathbf{R}}_i^{\top}\dot{\mathbf{p}}_j = \overline{\mathbf{R}}_i^{\top}\overline{\mathbf{R}}_j\,\mathbf{\overline{v}}_j = \mathbf{R}_z(\psi_j-\psi_i)\,\mathbf{\overline{v}}_j.
\label{eq:app_RiT_pjdot}
\end{equation}
Since $\alpha_{ij}=\psi_j-\psi_i-\phi_{ij}$, then $\psi_j-\psi_i=\alpha_{ij}+\phi_{ij}$. Therefore, \eqref{eq:app_RiT_pjdot} can be written as
\begin{equation}
\overline{\mathbf{R}}_i^{\top}\dot{\mathbf{p}}_j = \mathbf{R}_z(\alpha_{ij}+\phi_{ij})\,\mathbf{\overline{v}}_j.
\label{eq:app_RiT_pjdot_alpha_phi}
\end{equation}

With the required rotational relations established, substituting \eqref{eq:app_rhat}--\eqref{eq:app_RiT_pjdot_alpha_phi} into the time derivative of $\mathbf{q}_{ij}$ yields
\begin{align}
\dot{\mathbf{q}}_{ij} = -[\overline{\omega}_{z,i} \mathbf{b}_3]_\times (\mathbf{q}_{ij}+d_i\,\mathbf{b}_1) + \mathbf{R}_z(\alpha_{ij}+\phi_{ij})\mathbf{\overline{v}}_j - \mathbf{\overline{v}}_i.
\label{eq:app_sdot_pre}
\end{align}
Using the identity $[\mathbf{b}_3]_\times[a, b, c]^\top=[-b, a, 0]^\top$, we obtain
\begin{equation}
-[\overline{\omega}_{z,i} \mathbf{b}_3]_\times (\mathbf{q}_{ij}+d_i\,\mathbf{b}_1) = \overline{\omega}_{z,i}
\begin{bmatrix}
y_{ij} \\
-(x_{ij}+d_i) \\
0
\end{bmatrix}.
\end{equation}
Therefore, the components of $\dot{\mathbf{q}}_{ij}$ are given by
\begin{equation}
\dot{\mathbf{q}}_{ij} = -\mathbf{\overline{v}}_i + \mathbf{R}_z(\alpha_{ij}+\phi_{ij})\mathbf{\overline{v}}_j + \overline{\omega}_{z,i}
\begin{bmatrix}
y_{ij} \\
-(x_{ij}+d_i) \\
0
\end{bmatrix}.
\label{eq:app_sdot_core}
\end{equation}

Having obtained the Cartesian relative velocity $\dot{\mathbf{q}}_{ij}$, we project it into spherical coordinates to derive the scalar kinematics. First, for the relative distance $L_{ij}\triangleq\|\mathbf{q}_{ij}\|=\sqrt{\mathbf{q}_{ij}^{\top}\mathbf{q}_{ij}}$, differentiating yields
\begin{equation}
\dot{L}_{ij} = \frac{\mathbf{q}_{ij}^{\top}\dot{\mathbf{q}}_{ij}}{\|\mathbf{q}_{ij}\|} = \left(\frac{\mathbf{q}_{ij}}{\|\mathbf{q}_{ij}\|}\right)^{\top}\dot{\mathbf{q}}_{ij},
\label{eq:app_Ldot_id}
\end{equation}
where the unit 3D-LFF direction vector is
\begin{equation}
\frac{\mathbf{q}_{ij}}{\|\mathbf{q}_{ij}\|} =
\begin{bmatrix}
c_{\xi_{ij}} c_{\phi_{ij}}\\
c_{\xi_{ij}} s_{\phi_{ij}} \\
s_{\xi_{ij}}
\end{bmatrix}.
\end{equation}

Substituting $x_{ij}=L_{ij} c_{\xi_{ij}} c_{\phi_{ij}}$, $y_{ij}=L_{ij} c_{\xi_{ij}} s_{\phi_{ij}}$, $z_{ij}=L_{ij} s_{\xi_{ij}}$, and \eqref{eq:app_sdot_core} into \eqref{eq:app_Ldot_id} yields the range rate:
\begin{equation}
\begin{aligned}
\dot{L}_{ij} ={}& - c_{\xi_{ij}} \big( c_{\phi_{ij}}\, \overline{v}_{x,i} + s_{\phi_{ij}}\, \overline{v}_{y,i} \big) - s_{\xi_{ij}}\, \overline{v}_{z,i} - d_i\,\overline{\omega}_{z,i}\, c_{\xi_{ij}} s_{\phi_{ij}}\\
&+ c_{\xi_{ij}} \Big( c_{\alpha_{ij}}\, \overline{v}_{x,j} - s_{\alpha_{ij}}\, \overline{v}_{y,j} \Big) + s_{\xi_{ij}}\, \overline{v}_{z,j}.
\end{aligned}
\label{eq:app_Ldot}
\end{equation}
The azimuth angle rate follows from $\phi_{ij}=\mathrm{atan2}(y_{ij},x_{ij})$, where $\dot{\phi}_{ij}=(x_{ij}\dot{y}_{ij}-y_{ij}\dot{x}_{ij})/({x_{ij}}^2+{y_{ij}}^2)$, giving
\begin{align}
\dot{\phi}_{ij} =& \frac{s_{\phi_{ij}}\, \overline{v}_{x,i} - c_{\phi_{ij}}\, \overline{v}_{y,i}}{L_{ij} c_{\xi_{ij}}} 
-\frac{d_i\, c_{\phi_{ij}}}{L_{ij} c_{\xi_{ij}}}\,\overline{\omega}_{z,i} \nonumber\\
& +\frac{s_{\alpha_{ij}}\, \overline{v}_{x,j} + c_{\alpha_{ij}}\, \overline{v}_{y,j}}{L_{ij} c_{\xi_{ij}}} -\overline{\omega}_{z,i}.
\label{eq:app_phidot}
\end{align}
The elevation angle rate follows from $s_{\xi_{ij}}=z_{ij}/L_{ij}$, yielding
\begin{equation}
\begin{aligned}
\dot{\xi}_{ij} =& \frac{s_{\xi_{ij}}\big(c_{\phi_{ij}}\, \overline{v}_{x,i} + s_{\phi_{ij}}\, \overline{v}_{y,i}\big) - c_{\xi_{ij}}\, \overline{v}_{z,i}}{L_{ij}} +\frac{d_i\, s_{\xi_{ij}} s_{\phi_{ij}}}{L_{ij}}\,\overline{\omega}_{z,i}\\
&-\frac{s_{\xi_{ij}} c_{\alpha_{ij}}}{L_{ij}}\, \overline{v}_{x,j}+\frac{s_{\xi_{ij}} s_{\alpha_{ij}}}{L_{ij}}\, \overline{v}_{y,j} +\frac{c_{\xi_{ij}}}{L_{ij}}\, \overline{v}_{z,j}.
\label{eq:app_xidot}
\end{aligned}
\end{equation}
Finally, the relative heading rate is
\begin{align}
\dot{\alpha}_{ij} =& \overline{\omega}_{z,j} -\frac{s_{\phi_{ij}}\, \overline{v}_{x,i} - c_{\phi_{ij}}\, \overline{v}_{y,i}}{L_{ij} c_{\xi_{ij}}} +\frac{d_i\, c_{\phi_{ij}}}{L_{ij} c_{\xi_{ij}}}\,\overline{\omega}_{z,i} \nonumber\\
&-\frac{s_{\alpha_{ij}}\, \overline{v}_{x,j} + c_{\alpha_{ij}}\, \overline{v}_{y,j}}{L_{ij} c_{\xi_{ij}}}.
\label{eq:app_alphadot}
\end{align}
\balance

To guarantee the absence of kinematic singularities, we must verify that any relative position within the frustum set, $\mathbf{q}_{ij} \in \mathcal{S}_{ij}^\mathbf{q}$, satisfies $L_{ij} > 0$ and $\cos(\xi_{ij}) > 0$. 
For any $\mathbf{q}_{ij} \in \mathcal{S}_{ij}^\mathbf{q}$, the depth constraint $x_{ij} \ge x_\mathrm{near} > 0$ immediately guarantees that the relative distance satisfies $L_{ij} = \|\mathbf{q}_{ij}\| \ge x_{ij} > 0$. 
Furthermore, because $x_{ij} > 0$, the vertical frustum inequalities yield $|z_{ij}/x_{ij}| \le \tan(\Theta/2)$. Evaluating the tangent of the elevation angle using the Cartesian components gives
\begin{align}
    |\tan(\xi_{ij})| = \frac{|z_{ij}|}{\sqrt{x_{ij}^2 + y_{ij}^2}} \le \left|\frac{z_{ij}}{x_{ij}}\right| \le \tan(\Theta/2).
\end{align}
Since the vertical FOV satisfies $\Theta \in (0, \pi)$, the elevation angle is bounded by $\xi_{ij} \in (-\pi/2, \pi/2)$, which guarantees that $\cos(\xi_{ij}) > 0$.

With these geometric bounds established, the analytical regularity of the kinematics follows naturally.
The right-hand side of the system \eqref{eq:app_Ldot}--\eqref{eq:app_alphadot} comprises smooth trigonometric functions, linear state mappings, and rational fractions. The only denominators present are $L_{ij}$ and $L_{ij} \cos(\xi_{ij})$.
Evaluating the matrix $\mathbf{G}(\mathbf{x}_{ij})$ in \eqref{eq:FG} yields the determinant 
\begin{align}
    \det(\mathbf{G}) = \frac{1}{L_{ij}^2 \cos(\xi_{ij})}.
\end{align} 
As shown above, for any relative position $\mathbf{q}_{ij} \in \mathcal{S}_{ij}^\mathbf{q}$, we have $L_{ij} > 0$ and $\cos(\xi_{ij}) > 0$. 
Therefore, all denominators are strictly bounded away from zero. 
Consequently, this non-vanishing property guarantees two critical conditions over $\mathcal{S}_{ij}^\mathbf{q}$: (i) the input matrix $\mathbf{G}$ remains non-singular, and (ii) the kinematic vector fields remain continuously differentiable ($C^1$), ensuring local Lipschitz continuity.

% \begin{remark}[Frustum admissible set implies LOS admissibility]\label{rem:S_subset_C}
% Let $\mathbf{q}_{ij}=[x_{c,ij},\,y_{c,ij},\,z_{c,ij}]^\top \in \mathcal{S}$.
% Since $x_{c,ij}\ge \lambda_{\min}>0$, the frustum inequalities in \eqref{eq:frustum-set} imply
% $\left|\frac{y_{c,ij}}{x_{c,ij}}\right| \le T_h$ and
% $\left|\frac{z_{c,ij}}{x_{c,ij}}\right| \le T_v$.
% Using $|\tan(\phi_{ij})|=\left|\frac{y_{c,ij}}{x_{c,ij}}\right|$ gives $|\phi_{ij}|\le \Phi/2$.
% Moreover,
% $|\tan(\xi_{ij})|
% =
% \left|\frac{z_{c,ij}}{\sqrt{x_{c,ij}^2+y_{c,ij}^2}}\right|
% \le
% \left|\frac{z_{c,ij}}{x_{c,ij}}\right|
% \le T_v,$
% which implies $|\xi_{ij}|\le \Theta/2$.
% Hence $\mathbf{x}_{ij}\in\mathcal{C}$ whenever $\mathbf{q}_{ij}\in\mathcal{S}$, i.e., $\mathcal{S}\subseteq\mathcal{C}$.
% Therefore, all results established on $\mathcal{C}$ also hold on $\mathcal{S}$.
% \end{remark}
\newpage

% your own custom macro containing the $\backslash${\tt{includegraphics}} command to make things
 

% \end{IEEEbiography}



\vfill

\end{document}